\documentclass[11pt, letterpaper, romanappendices, onecolumn]{article}

\usepackage[T1]{fontenc}
\usepackage[cmex10]{amsmath}
\usepackage{amsfonts}
\usepackage{amssymb}
\usepackage{amsthm}
\usepackage{graphicx}
\usepackage{lmodern}
\usepackage{natbib}
\usepackage{hyperref}
\usepackage{url}
\usepackage{color}
\usepackage{fullpage}

\theoremstyle{plain}% default
\newtheorem{thm}{Theorem}[section]

\newtheorem{prop}[thm]{Proposition}

\theoremstyle{definition}
\newtheorem{defn}{Definition}[section]

\theoremstyle{remark}

\title{Some Efficient Solutions to Yao's Millionaire Problem}
\author{%
	Ashish Kumar\footnote{Computer Science and Engineering,
	Indian Institute of Technology Jodhpur, India. Email: \texttt{ashishkumar@iitj.ac.in}},
	Anupam Gupta\footnote{Center of Excellence -- ICT, Indian Institute of Technology Jodhpur, India. Email: \texttt{ag@iitj.ac.in}}
}

\begin{document}
	\maketitle
\begin{abstract}
We present three simple and efficient protocol constructions to solve Yao's Millionaire Problem when the parties involved are non-colluding and semi-honest. The first construction uses a partially homomorphic Encryption Scheme and is a $4$-round scheme using $2$ encryptions, $2$ homomorphic circuit evaluations (subtraction and XOR) and a single decryption. The second construction uses an untrusted third party and achieves a communication overhead linear in input bit-size with the help of an order preserving function.Moreover, the second construction does not require an apriori input bound and can work on inputs of different bit-sizes. The third construction does not use a third party and, even though, it has a quadratic communication overhead, it is a fairly simple construction.
\end{abstract}

\section{Introduction}
Alice (owning a private variable $x$) and Bob (owning a private variable $y$) wish to determine the truth value of the predicate ``$x > y$'' without disclosing any of their private data except for what is implied by the result. This article details the attempts to solve this problem efficiently by using techniques different from the ones already presented. This problem is famously known as \textit{the Millionaire Problem} in the literature. The solution to the millionaire problem finds applications in e-commerce and data mining; and also forms a sub-procedure in solving \textit{the vector dominance problem} \cite{du2001study}.

\par The remainder of the article discusses the related research in this area (Section \ref{sec2}) and presents three new constructions -- termed `\textsf{A}', `\textsf{B}' and `\textsf{C}' (Section \ref{sec3}). Each protocol starts with definitions, terminologies and constructions used, followed by the protocol, its complexity and a brief security analysis. Any other added details in a protocols are included thereafter. Finally, in Section \ref{sec4}, we conclude and mention some directions that can be taken to further this research.

\section{Background and Related Work} \label{sec2}
\citet{yao1982protocols} gave the first protocol for solving the secure comparison problem. However, the solution was exponential in time and space requirements. In 1987, \citet{goldreich1987play} proposed a solution using scrambled circuits to any secure multiparty computation problem.

\par \citet{beaver1990round} presented a constant-round solution for multi-party computation. \citet{crypto19871119},  \citet{beaver1990multiparty}, and \citet{goldwasser1991fair} have studied fairness in two-party computations with extension to multiparty computation. \citet{feige1994minimal} study the multi-party secure computation models in the presence of an untrusted third party which does not learn anything about the inputs. \citet{cachin1999efficient} also presented an elegant and practical solution to the millionaire problem using an untrusted third party $T$ based on the $\Phi$-hiding assumption.

\par \citet{schoenmakers2004practical} used \textit{threshold homomorphic encryption schemes} to solve the multiparty computation problem. \citet{lin2005efficient} proposed a two-round protocol for solving the millionaire problem in the setting of semi-honest parties using \textit{multiplicative or additive} homomorphic encryption schemes. \citet{fischlin2001cost} constructed a two-round protocol for the millionaire problem using the Goldwasser-Micali encryption scheme (The computation cost of the protocol is $O(\lambda~n)$ modular multiplications, where $\lambda$ is the security parameter and the communication cost is $O(\lambda~n \log N)$, where $N$ is the modulus). \citet{blake2004strong} presented a two-round protocol for the problem using the \textit{additive homomorphic Paillier cryptosystem}. Its computation cost is $O(n \log N)$ and the communication cost is $O(n \log N)$. A Symmetric cryptographic solution to the millionaires problem and evaluation of secure multiparty computations was presented by \citet{shundong2008symmetric}. In 2003, \citet{ioannidis2003efficient} proposed an efficient protocol for the problem, having a suboptimal time and communication complexity. In 2009, \citet{gentry2009fully} proposed the first fully homomorphic encryption scheme which allows one to compute arbitrary functions over encrypted data.

\section{Construction of the Protocols}\label{sec3}
We work out three protocol constructions that are efficient and simple.

\subsection{Protocol \textsf{A} Construction}
\textsc{Assumption:} We assume the existence of an efficient partially homomorphic encryption. By partial, we mean that the homomorphic encryption should have the properties of additivity and bit-wise XOR.

\par Alice owns a private value $a$ and Bob owns the private value $b$. Here we represent $a$ and $b$ in two's compliment form such that the sign of any integer (represented as a binary string) is stored as the most significant bit of its corresponding binary string. Also, all XOR operations performed are bitwise.

\begin{enumerate}
	\item Bob owns a Homomorphic public key encryption scheme $(E, D)$.
	\item Bob now computes $E(b)$ and sends it to Alice.
	\item Alice generates a random number $R$ and computes:
	\begin{enumerate}
		\item $E(a)$,
		\item then $V = (E(a)-E(b)) \oplus E(R)$
	\end{enumerate}
	and sends $V$ to Bob. Note that $V = E((a-b)\oplus R)$ by the properties of homomorphic encryption.
	\item Bob decrypts $V$ to obtain $(a-b) \oplus R$ and sends the most significant bit (MSB) of the decrypted value to Alice -- it contains the information about the sign of the operation $(a - b)$.
	\item Alice then takes the XOR of the obtained bit from Bob and the MSB of $R$ to obtain the output: if $a > b$.
\end{enumerate}

\subsubsection{Analysis}
\par The security analysis of Protocol \textsf{A} is trivial and is based on the security of the corresponding homomorphic encryption scheme and the fact that one time pads are secure for a single use of a key.

\par \textsc{Communication Overhead}: Let the input with larger number of bits be $a$. Then the communication overhead can be seen to be: $2|E(a)| + 2$.

\par \textsc{Computation}: The protocol is efficient. The computation overhead is: $2$(Complexity of encrypting a variable) + (Complexity of subtraction in the homomorphic scheme) + (Complexity of XOR in the homomorphic scheme) + (Complexity of XOR of 2 bits (step 5 of the protocol)). The first term depends on the encryption scheme used; the second and third term depend on the homomorphic scheme used and the last term is a constant.

\subsection{Protocol \textsf{B} Construction}
As we will formalize the protocol, we will also note that even if a single party, out of the two involved parties, is computationally powerful, the protocol will be able to exploit the computation power of that party to perform the computationally expensive tasks in the protocol.
\par We begin with the following definition.
\begin{defn}
(\textbf{Order Preserving Function:}) A function $F: A\rightarrow B$ (where $|A| \ll |B|$) is called an order preserving function if $a < b$ $\Longleftrightarrow F(a) < F(b)$, $\forall a, b \in A$.
\end{defn}
For example, $f(x) = 2x$ is, clearly, an order preserving function.
\newline

\subsubsection{General Idea of Protocol \textsf{B}}
\par Suppose Alice and Bob have some random order preserving function $F: A\rightarrow B$, unknown to a third party. Later we'll see how to construct such an order preserving function.

\par Now if Alice and Bob send $X = F(a)$ and $Y = F(b)$ to Ursula (the third party), then $a > b$ \textit{iff} $X > Y$ -- follows from the definition of order preserving function.
\newline

\subsubsection{Construction of a Random Order Preserving Function, $F$}
\par Let the input size be $n$ bits, i.e., $n = \max(\lceil \log a \rceil, \lceil \log b \rceil)$. As we will present the construction, we will note that that construction need not know the input sizes in advance and neither does it
require the inputs to have the same bit-sizes.

\par Then to construct $F$, we consider $n$ random functions: $f_1, f_2, \ldots f_n$, where each $f_i: \lbrace 0,1\rbrace \rightarrow\mathbb{N}$ is used to encode the $i^\text{th}$ least significant input bit.

Let us assume that the mappings follow the following constraints:
\begin{eqnarray}
(f_i(1)- f_i(0)) &>& \displaystyle\sum\limits^{i-1}_{j=1} (f_j(1) - f_j(0)),\\
f_i(1) &>& f_i(0) \qquad \forall i.
\end{eqnarray}

\par We can always find such values for $f_1, f_2, f_3 \ldots f_n$ satisfying the above constraints. For a sample construction of $F$, refer to Section \ref{excons}
\par Let $a_i$, for $1 \le i \le n$, denote the $i^\text{th}$ least significant bit in the binary representation of an input $a$. Then define
\begin{equation}
F(a) := f_n(a_n) + f_{n-1}(a_{n-1}) + \ldots + f_2(a_2) + f_1(a_1),
\end{equation}
where, $F:\lbrace 0,1 \rbrace^n \rightarrow\mathbb{N}$.

\begin{thm}\label{ope}
$F$ is an order preserving function.
\end{thm}
\begin{proof}
Let $x_0,y_0\in\mathbb{N}$. Assume, without the loss of generality, $x_0 > y_0$. Let $x,y$ correspond to the binary representation of $x_0,y_0$ respectively. Here $x,y \in \lbrace0,1\rbrace^n$. We show that $F(x) > F(y)$, where $F$ is constructed as above. We denote the $i^\text{th}$ bit (starting from the least significant bit numbered $1$) of a bit string $x$ by $x_i$. Since $x > y$, $x$ must have a $1$ at the first bit position at which $x$ and $y$ \textit{differ}, starting from the most significant bit. Let this position be $p$. By the construction of the function $F$, we have:
\begin{equation}
(f_p(1)- f_p(0)) > \displaystyle\sum\limits^{p-1}_{j=1}(f_j(1)-f_j(0)).	\label{by_cons}
\end{equation}

\par Now,
\[
	F(x) = f_n(x_n) + f_{n-1}(x_{n-1}) + \ldots + f_1(x_1);
\]
and 
\[
	F(y) = f_n(y_n) + f_{n-1}(y_{n-1}) + \ldots + f_1(y_1).
\]

So,
\[
	F(x) - F(y) = f_p(1)-f_p(0) + \displaystyle\sum\limits_{j=1}^{p-1}(f_j(x_j) - f_j(y_j)),
\]
Since, all bits at positions higher than $p$ are equal for $x$ and $y$ and hence, their mappings being equal, cancel out.

\par Thus for $F(x) > F(y)$, we require
\begin{eqnarray}
	f_p(1)-f_p(0) &+& \displaystyle\sum\limits_{j=1}^{p-1}(f_j(x_j) - f_j(y_j)) > 0 \\
	\Rightarrow \qquad f_p(1)-f_p(0) &>& \displaystyle\sum\limits_{j=1}^{p-1}(f_j(y_j) - f_j(x_j)). \label{eq:3.1.6}
\end{eqnarray}

To show that inequality \eqref{eq:3.1.6} holds, we show that it holds for the maximum value of the RHS. RHS is maximum when we have $x_k = 0,~\forall~k < p$ and $y_k = 1,~\forall~k < p$ (this is because if we have $x_k=y_k=0$ or $x_k=y_k=1$, then $f_k(y_k)-f_k(x_k)=0$, and since $f_k(1) > f_k(0) ~\forall~k$, by construction; hence for $x_k=0,y_k=1$, we have $f_k(y_k)-f_k(x_k)>0$, i.e. it has a positive contribution to the sum on the R.H.S). So, we require that:
\[
f_p(1)- f_p(0) > \displaystyle\sum\limits^{p-1}_{j=1}(f_j(1)-f_j(0)),
\]
which we know is true by the construction of $F$ (see equation \eqref{by_cons}).

The proof for $x<y$ $\Longrightarrow$ $F(x)<F(y)$ is analogous. This completes the proof. 
\end{proof}

\subsubsection{An Example Construction of $F$} \label{excons}
Let us consider an order preserving function $F: \lbrace 0,1 \rbrace^4 \rightarrow \mathbb{N}$.

\par An order preserving function $F$ can be thought as a sequence of functions mapping each bit position, i.e., 
\[
	F = \lbrace f_{1} , f_{2}, f_{3}, f_{4}\rbrace
\]
where $f_i$ represents the mapping corresponding to the $i^\text{th}$ bit starting from the least significant bit (which is the first bit.)
\newline

\par While constructing the functions $f_i$, we need to take care of the following constraints:
\begin{eqnarray*}
	f_{i}(1) - f_{i}(0) &>& \displaystyle\sum\limits^{i-1}_{j=1}(f_{j}(1)-f_{j}(0)),\\
	f_{i}(1) &>& f_{i}(0) \qquad  \forall i.
\end{eqnarray*}

\par The computation of $f_i$s is tabulated in Table \ref{tab:1}

\begin{table}
	\begin{center}
	\begin{tabular}{|c|c|c|c|c|}
	\hline 
		$f_{i}$ & $f_{i}(0)$ & $f_{i}(1)$ & $f_{i}(1) - f_{i}(0)$ & $\sum\limits_{j=1}^{i} f_{j}(1)-f_{j}(0)$ \\\hline 
		$f_{1}$ & 3 & 5 & 2 & 2 \\\hline
		$f_{2}$ & 7 & 10 & 3 & 5 \\\hline
		$f_{3}$ & 1 & 8 & 7 & 12 \\\hline
		$f_{4}$ & 4 & 18 & 14 & 26 \\\hline
	\end{tabular}
	\end{center}
	\caption{Computation of $f_{i}$} \label{tab:1}
\end{table}

Table \ref{tab:2} tabulates the image of all 4 bit numbers using the defined mapping $F$. 

\begin{table}
	\begin{center}
	\begin{tabular}{|c|c|}\hline 
		Input ($x$) & $F(x)$ \\\hline 
		0000 & 15 \\\hline 
		0001 & 17 \\\hline
		0010 & 18 \\\hline 
		0011 & 20 \\\hline
		0100 & 22 \\\hline
		0101 & 24 \\\hline
		0110 & 25 \\\hline
		0111 & 27 \\\hline
		1000 & 29 \\\hline
		1001 & 31 \\\hline
		1010 & 32 \\\hline
		1011 & 34 \\\hline
		1100 & 36 \\\hline
		1101 & 38 \\\hline
		1110 & 39 \\\hline
		1111 & 41 \\\hline 
	\end{tabular} 
	\end{center}
	\caption{Computing $F$} \label{tab:2}
\end{table}

\par \textbf{NOTE:} While constructing the function $F$, we can independently choose $f_{i}(0)$ for every $i$ and then we can accordingly choose $f_{i}(1)$ to satisfy the constraints mentioned.

\subsubsection{Constructing Shared Function} \label{shared}
\par Consider an order preserving function $F$. We will show one possible construction of $F$ and argue its correctness thereafter.
Here, $F(x)=F(x_n\ldots x_1)$ is defined as:
\[
	F(x_n\ldots x_1) = f_n(x_n) + f_{n-1}(x_{n-1}) + \ldots +f_2(x_2) + f_1(x_1),
\]
\par where the values of $f_i(0)$ and $f_i(1)$ are:
\[
	f_i(0)=s \quad \forall i
\]
\[
	f_i(1)=s+k^il \quad \forall i
\]
\par \quad where $s,k (>1)$ and $l$ are randomly chosen constants, whose choice is critical to the protocol \ref{protb}.

\begin{prop}
$F$ is an Order Preserving function.
\end{prop}

\begin{proof}

\par By theorem \ref{ope}, we only need to show that
\[
	f_i(1)- f_i(0) > \sum\limits^{i-1}_{j=1}(f_j(1)-f_j(0))
\]
or,
\begin{equation}
	f_i(1) > \sum\limits^{i-1}_{j=1}(f_j(1)-f_j(0)) + f_i(0). \label{eq:3.2.7}
\end{equation}

\par We first show how the encryption handles inputs of different bit sizes, provided that they do not contain any leading zeros (i.e. 00101 should only be represented as 101). Without loss of genrality, we can assume that $|a| > |b|$. Then the most significant bit of $a$ (which is 1) has position $i (> j)$, where j is the most significant bit-position of $b$. For the construction to work, we require:
\[
	F(a)>F(b)
\] 
or

\[
	\sum\limits^i_{x=1}f_x(a_x)>\sum\limits^j_{x=1}f_x(a_x)
\] 

\par We show this by plugging the maximum value of RHS and comparing it with the minimum value of LHS (i.e. we take $a=10 \ldots 0$ and $b=11\ldots 1$):

\[
	f_i(1)+\sum\limits^{i-1}_{x=1}f_x(0)>\sum\limits^j_{x=1}f_x(1)
\] 
 
\[
	si + k^il>\sum\limits^j_{x=1}(s+k^xl)= sj + \frac{(k^{j+1}-k)l}{k-1}
\] 
Which we can easily see by combining the following two inequalities:
\[
	si>sj, \quad\text{Since i>j}
\]
And,
\[
	k^il \geq k^{j+1}l, \quad\text{As i>j}
\]

\[
	\Longrightarrow k^il>(k^{j+1}l-kl)>\frac{(k^{j+1}-k)l}{k-1} \quad\text{Since k>1}. 
\]

\par We now prove the same for equal bit sizes (where we assume that both the inputs $-$ $a$ and $b$ differ at bit-position $i$, when we start scaning from the MSB and since we have assumed that $a>b$, $a$ has a 1 at position $i$ and $b$ has a 0 at position $i$) and show that inequality \eqref{eq:3.2.7} holds for the maximum value of the RHS (i.e. all bit-positions of $a$ less significant than $i^{th}$ are 0 and all bit-positions of $b$ less significant than $i^{th}$ are 1 ). So, plugging in the values in inequality \eqref{eq:3.2.7}, we require

\[
	f_i(1) > \sum\limits^{i-1}_{j=1}(s+k^j l- s) + s,
\]

\begin{equation}
	f_i(1) > s+\frac{(k^i-k)l}{k-1} \label{eq:3.2.8}
\end{equation}

\par But,

\[
	f_i(1) = s+k^i l > s+\frac{(k^i-k)l}{k-1}, \quad\text{(In our construction)}
\]
\par \quad where the second inequality follows from the fact that $k>1$.

\par Thus the required inequality \eqref{eq:3.2.8} is satisfied by our construction.

\par Since the above argument is true for all $i$, we have shown that $F$ is an Order Preserving function.
\end{proof}

\par\textbf{ Note: Our construction works well even if the bit-size of the inputs are unequal and not known apriori. We also do not require any bound on input for the function since the function is defined for bit positions upto $\infty$.}

\subsubsection{Formalizing Protocol \textsf{B}} \label{protb}
\begin{enumerate}
	\item Alice generates $l, s, k$ $-$ the parameters for construction of $F$, and communicates it to Bob.
	\item Alice then flips an unbiased coin and obtains: bit $u \in \lbrace 0,1\rbrace$, and sends it to Bob.
	\item If $u=1$, they use their inputs as is, otherwise, if $u=0$, they take the bit-wise complement of their inputs and use the complemented inputs in the protocol instead of the original inputs, thus reversing the order of their inputs.
	\item Alice and Bob now construct the function $F$ using the common parameters: $s,l,k,u$ and generate $F(a)=A$ and $F(b)=B$; where $a$ is the binary representation of Alice's private input and $b$ is the binary reprensation of Bob's private input (or the bitwise complements of their inputs, depending on the value of $u$); and send it to Ursula.
	\item Ursula compares $A$ and $B$ and tells the result to Alice and Bob, where Alice and Bob obtain the protocol outcome depending on the value of the unbiased bit $u$. If $u=0$, they reverse the result obtained by Ursula to get the correct outcome; else they take the outcome to be what Ursula tells them.
\end{enumerate}

The correctness of the protocol follows from previous discussions.

\textbf{Note: Here, as wee can see that all the computationally expensive tasks are being done by Alice, thus, in real time applications, when using the protocol with one of the two parties being a resourceful server, we can assign the role of Alice to that server and use its computation power to perform all the computationally expensive tasks, making the protocol even more computationally efficient and practical.}

\subsubsection{Analysis}
\par Note that
\[
	\max \lbrace F(x) \rbrace = \sum \limits_{j=1}^{n} s+k^i l = ns + \frac{(k^{n+1}-k)l}{k-1},
\]
for $x \in \lbrace 0,1 \rbrace^n$.

\par Thus, maximum number of bits in $F(x) = log(ns + \frac{(k^{n+1}-k)l}{k-1})=O(nlogk + logl + logs + logn)$. Thus size of $F(x)$ is linear in $n$ and requires $O(n)$ bits to communicate it to the untrusted third party (where $k,l,s$ are randomly chosen constants).

\par \textsc{Communication Overhead:} The protocol requires communicating: $s, l, k, u$ to Bob (by Alice) ($O(logs+logl+logk+1)$ bits); $F(x)$ and $F(y)$ (both $O(n)$ bit long) to Ursula; and sending the result back to Alice and Bob, using a single bit for each. Hence, the overall overhead is linear in $n$.

\par \textsc{Computational Overhead:} $F(x)$ is $O(n)$ time computable. Thus the overall computational overhead is $O(n)$ + (computational overhead of generating an unbiased bit).

\par \textsc{Security:}\label{issue}
\begin{enumerate}
	\item The third party doesn't know the value of $s, l$ or $k$, and the only value that Ursula obtains is $F(a)$ and $F(b)$. Here, $F(a)$ is of the form:

\[
	F(a)=\sum \limits_{j\in U}s + \sum \limits_{j\in V}(s+k^jl)=ns + \sum \limits_{j\in V}k^jl;
\]
where $U=\{i:a_i=0\}$ and $V=\{i:a_i=1\}$. 
\par Thus, the securty of individually guessing $F(a)$ or $F(b)$ depends on the security of $n$, which is unknown since protocol does not require input number of bits apriori and neither does it require them to be equal; and also depends on the hardness of factoring, as the algorithm would require the factoring of the second summand to obtain $k,l$; which are chosen to be large to make the protocol secure. Even before that, Ursula would need to guess $ns$ from $F(a)$, for which there are exponential number of possibilities (because $F(a)$ is $O(n)$ bits and thus its value can be $0$ to $2^n$ bits long) and hence is exponentially hard to guess. Another information that Ursula can obtain is $F(a)-F(b)$:
\[
	F(a)-F(b)=\sum \limits_{j\in V_a}k^jl-\sum \limits_{j\in V_b}k^jl= constant \times l;
\]
The security of $F(a)-F(b)$ is dependant on the difficulty of factoring of the R.H.S value to obtain $l$, which would be difficult to guess if $l$, $k$ (which determines the $constant$) are large composite numbers making it difficult to factor R.H.S at first place and on top of that, guess which factor of the value equals $l$.
	\item Since the values sent to the third party may be order preserved or order reversed with equal probability, the probability distribution is also evenly distributed for lower as well as higher values. The third party also cannot gain information about the minimum or maximum value of Alice's or Bob's private input (which could possibly have been given away by $F(x)$), since the event that $\{x > F(x)\}$ and the event $\{x < F(x)\}$ are equally likely in general and in our protocol as well (since the bit $u$ is unbiased).
	\item Since $F(x)$ and $F(y)$ have atleast as many elements between them as $x$ and $y$ ($F$ is an order preserving function), this gives away information regarding the maximum gap between $x$ and $y$ (via the difference between $F(x)$ and $F(y)$). Although this data is not statistically useful in many scenarios, we can still patch it with an added communication overhead as shown later in the extension to this construction in Section \ref{sec:ext}.
\end{enumerate}

\subsection{Extension of Protocol \textsf{B}}\label{sec:ext}
The construction of $F$ in section \ref{shared} can be modified to make $k$ an $O(n)$-bit long value so that the gap between the mapping of the numbers under $F$ which differ only in their least significant bit (this gap corresponds to the minimum achievable gap between two input numbers) has a gap of atleast $2^n$ and hence gives away no relevant statistical data about the input ( Since for n-bit numbers, knowing that there are $2^n$ elements between them does not enable us to narrow down on any value or a range of values).This resolves the issue mentioned earlier in Section \ref{issue}.

\par Clearly, the maximum number of bits to be transmitted in this case is:
\[
	O(log(\max \lbrace F(x) \rbrace)) = O(log(\sum \limits_{j=1}^{n} s+k^i l)) =O(ns + \frac{(k^{n+1}-k)l}{k-1})=O(nlogn),
\]
for $k,l,s=O(n)$.

\par Thus, this new protocol has a linearithmic communication overhead. Its computation overhead is same as before.

\subsection{Protocol \textsf{C} Construction}
\begin{defn}
\textbf{Order Preserving Function at a point $b$:} A function $F: A \rightarrow B$ is called an order preserving function at point $b$ if for every $x \in A$, $x < b\Longrightarrow F(x) < F(b)$ and $x > b \Longrightarrow F(x) > F(b)$. 
\end{defn}

\par Here we present one possible construction of an \textit{Order Preserving Function at a point $b$}. 

\par The function $F$ generates the output of its input $a$ bit-wise, with each input bit $a_i$ (where $a_i$ represents the $i^{th}$ bit starting from the least significant bit numbered 1, in the binary representation of the input $a$) being mapped to its output by a random function: $f_i : \lbrace 0,1 \rbrace \longrightarrow \mathbb{N}$.
\begin{defn}\label{definerise}
We define a set $S=\lbrace i: b_i = 0 \rbrace$, ($b_i$ denotes the $i^\text{th}$ least significant bit in the binary representation of the value  $b$). Then, we define a $\mathbf{(rise)_i}$ to be $[f_i (1) -f_i(0)], \quad \forall i \in S$. 
\end{defn}

\par \textbf{Note:} $\mathbf{(rise)}$ refers only to the elements in $S$.
 
\begin{defn}\label{definefall}
Similarly, we define a set $Y = \lbrace i:b_i =1 \rbrace$, ($b_i$ denotes the $i^\text{th}$ least significant bit in the binary representation of the value  $b$). Then, we define a $\mathbf{(fall)_i}$ to be $[f_i (1) -f_i(0)], \quad \forall i \in Y$.
\end{defn}

\par \textbf{Note:} $\mathbf{(fall)}$ refers only to the elements in $Y$.

\subsubsection{Construction of $F$}\label{ConsC}
\begin{enumerate}
	\item We randomly map $f_i (b_i)$ to an element in $\mathbb{N}$, where $b_i$ corresponds to the $i^{th}$ least significant bit in the binary representation of the input $b$.
	\item Now, $\forall i \in Y$, we randomly select a mapping for $f_i (0)$ in $(f_i (1) - \mathbf{v_i} - l, f_i (1)-\mathbf{v_i})$, where $l$ is a parameter randomly chosen by the constructor of the function, and
	\[
		\mathbf{v_i} = \sum\limits^{}_{\forall j \in S \text{ s.t. } j < i} \mathbf{(rise)_j}. \label{eqn rise}
	\]
	\item Similarly, $\forall i \in S$, we randomly select a mapping for $f_i (1)$ in $(f_i (0) + \mathbf{u_i}, f_i(0)+ \mathbf{u_i} + l)$; (where $l$ is the parameter chosen in step 2), and
	\[
		\mathbf{u_i} = \sum\limits^{}_{\forall j \in Y \text{ s.t. } j < i} \mathbf{(fall)_j}.
	\]
	\item We have our construction ready.
	\[
		F(x) = F(x_n\ldots x_1)=f_n (x_n) + f_{n-1} (x_{n-1}) + \cdots + f_1(x_1).
	\]
\end{enumerate}

\subsubsection{Proof of Correctness}
\par Here we only prove that for all $x\in \mathbb{N}$, $x < b\Longrightarrow F(x) < F(b)$, where $F$ is an order preserving function at $b$. The proof for ``$x > b \Rightarrow F(x) > F(b)$'' is analogous.

\par If $x < b$, then at the first bit (starting from MSB) at which $x$ and $b$ differ, $x$ must have a $0$ and $b$ must have a $1$. Let this position be $t$ (w.r.t. the position of the LSB, which is 1). Clearly $t\in Y$, where $Y$ is defined in definition \ref{definefall} 

\par The critical case or the worst case, proving which will prove all other cases, is the one where all bits of $x$ which are at positions less significant than $t$ equal $1$, i.e. $x_i=1 \quad \forall i<t$. Also, the bit-positions more significant than position $t$ are the same in both $x$ and $b$, since $t$ is chosen to be the first bit-position at which $x$ and $b$ differ.Thus,the only positions at which $x$ differs from $b$ are the positions less significant than $t^\text{th}$ bit at which bits of $b$ are $0$ and bits of $x$ are $1$ and all such positions correspond to a ``rise'', forming a subset of $S$ [$S$ is defined in definition \ref{definerise}](since other bits in $x$ and $b$ which are less significant than $t^\text{th}$ bit are both 1 and hence identical).

\par By construction, we have
\begin{flushleft}

$\mathbf{(fall)_t}=f_t(1)-f_t(0)$

$\Longrightarrow\mathbf{(fall)_t}> f_t(1)-(f_t(1)-v_t)=v_t $,\quad plugging in the maximum value of $f_t(0)$ from constraint \ref{eqn rise}

$\Longrightarrow\mathbf{(fall)_t}> \sum\limits^{}_{\forall j \in S, j<t}\mathbf{(rise)_j}$, \quad By deinition of $v_i$. [refer \ref{eqn rise}].

$\Longrightarrow f_t(1) - f_t (0) > \sum\limits^{}_{\forall j \in S, j<t} (f_j(1) - f_j(0))$, \quad By definition of \textbf{(rise)$_j$} [refer \ref{definerise}]

$\Longrightarrow f_t(1) + \sum\limits^{}_{\forall j \in S, j<t} f_j(0) > f_t(0) + \sum\limits^{}_{\forall j \in S, j<t} f_j(1)$, \quad By rearranging the terms

\par Let $G=\{i: i>t$ or $(i<t$ and $i\in Y)\}$. These are precisey the positions not covered in the above inequality. Since all these bit-positions are identical in the binary representation of both $x$ and $b$, i.e. $G$ also equals the set: $\{i:b_i=x_i\}$, thus

$	\sum\limits^{}_{j\in G}f_j(b_j)=\sum\limits^{}_{j\in G}f_j(_j)$

Adding the above equality to the previous inequality, we get

$f_t(1) + \sum\limits^{}_{\forall j \in S, j<t} f_j(0)+\sum\limits^{}_{j\in G}f_j(b_j) > f_t(0) + \sum\limits^{}_{\forall j \in S s.t.j<t} f_j(1) + \sum\limits^{}_{j\in G}f_j(x_j)$

$\Longrightarrow \sum\limits_{\forall j}f_j(b_j) > \sum\limits_{\forall j}f_j(x_j)$, 

$\Longrightarrow F(b) > F(x)$, 

\end{flushleft}

Thus, our proof is complete.

\subsubsection{Sample Construction}
\par Let the input bound on bit size be $4$. Let the function $F$ be constructed w.r.t. $1001$ as per the construction procedure in Sec \ref{ConsC} . We have the following table
\begin{center}
	\begin{tabular}{|c|c|c|c|c|}\hline
	$f_i$ & $f_i(0)$ & $f_i(1)$ & $\mathbf{(rise)_i}$ & $\mathbf{(fall)_i}$\\\hline
	$f_0$ & 13 & \textbf{25} & - & 12 \\\hline
	$f_1$ & \textbf{36} & 54 & 18 & - \\\hline
	$f_2$ & \textbf{30} & 43 & 13 & - \\\hline
	$f_3$ & -32 & \textbf{1} & - & 33 \\\hline 
	\end{tabular}
\end{center}
The values in boldface represent the values which have been chosen uniformly at random (since they correspond to the bits of $1001$).

\par We thus have the following table for the function $F$.
\begin{center}
\begin{tabular}{|c|c|}
\hline 
$x$ & $F(x)$ \\ 
\hline 
0000 & 47 \\ 
\hline 
0001 & 59 \\ 
\hline 
0010 & 65 \\ 
\hline 
0011 & 77 \\ 
\hline 
0100 & 60 \\ 
\hline 
0101 & 72 \\ 
\hline 
0110 & 78 \\ 
\hline 
0111 & 90 \\ 
\hline 
1000 & 80 \\ 
\hline 
\textbf{1001} &\textbf{ 92} \\ 
\hline 
1010 & 98 \\ 
\hline 
1011 & 110 \\ 
\hline 
1100 & 93 \\ 
\hline 
1101 & 105 \\ 
\hline 
1110 & 111 \\ 
\hline 
1111 & 123 \\ 
\hline 
\end{tabular}
\end{center}
Clearly, the function $F$ preserves order at point $b = 1001$ (in boldface). 

\subsubsection{Worst case analysis of $\mid$F(x)$\mid$}
\par  Let $b=b_d b_{d-1}\ldots b_2 b_1$. Let $f_i (b_i),\quad \forall i$ be randomly chosen in the range $[n, m]$. For the worst case analysis, we choose the maximum value for $f_i(0)$ $\forall i \in \mathbf{S}$ and minimum value for $f_i(1)$ $\forall i \in \mathbf{Y}$ from the ranges specified in the construction \ref{ConsC}. Let $| b | = d$. We have the following table which analyzes the maximum number of bits required to encode $F(x)$ for any $x$ (with $x_i$ denoting the $i^{th}$ least significant bit of $x$). For the worst case analysis, we take b=$101010101 \ldots 010$. We first tabulate the data, and explanation follows.
\begin{center}
\begin{tabular}{|c|c|c|c|c|c|c|}\hline
 $f_i$ & $f_i(0)$ & $f_i(1)$ & $\mathbf{(rise)_i}$ & $\mathbf{(fall)_i}$ & $\sum\limits_{\text{upto } i} \mathbf{(rise)_j}$ & $\sum\limits_{\text{upto } i} \mathbf{(rise)_j}$ \\\hline 
 $f_1$ & $m$ & $m+l$ & $l$ & - & $l$ & - \\\hline 
 $f_2$ & $n-l$ & $n$ & - & $l$ & $l$ & $l$ \\\hline 
 $f_3$ & $m$ & $m+2l$ & $2l$ & - & $3l$ & $l$ \\\hline 
 $f_4$ & $n-2l$ & $n$ & - & $2l$ & $3l$ & $3l$ \\\hline 
 $f_5$ & $m$ & $m+4l$ & $4l$ & - & $7l$ & $3l$ \\\hline 
 $f_6$ & $n-4l$ & $n$ & - & $4l$ & $7l$ & $7l$ \\\hline 
 $\vdots$ & $\vdots$ & $\vdots$ & $\vdots$ & $\vdots$ & $\vdots$ & $\vdots$ \\\hline 
 $f_j$ & $m$ & $m+2^{(j-1)/2}l$ & $2^{(j-1)/2}l$ & - & $(2^{(j+1)/2}-1)l$ & $(2^{(j-1)/2}-1)l$ \\\hline 
 $f_{j+1}$ & $n-2^{(j-1)/2}l$ & $n$ & - & $2^{(j-1)/2}l$ & $(2^{(j+1)/2}-1)l$ & $(2^{(j+1)/2}-1)l$ \\\hline 
 $\vdots$ & $\vdots$ & $\vdots$ & $\vdots$ & $\vdots$ & $\vdots$ & $\vdots$ \\\hline 
 $f_d$ & $n-2^{(d-1)/2}l$ & $n$ & - & $2^{(d-1)/2}l$ & $(2^{(d+1)/2}-1)l$ & $(2^{(d+1)/2}-1)l$ \\\hline 
\end{tabular}
\end{center}

\par For the worst case, \\
$f_i(b_i) = n, \quad \forall i \in Y$ and, \\
$f_i(b_i) = m, \quad \forall i \in S$ .

\par Clearly, the worst case number of bits required to represent $F(x)$ (corresponding to $F(1111...1)$) is $O(\log (d/2(m+n)+ \sum\limits_{j=1}^{j= \lfloor d/2 \rfloor} 2^j l)) = O(\log (2^{\lfloor d/2\rfloor+1} + d/2(m+n))) = O(d/2+\log (d/2)) = O(d)$. \par Also note that, since the function preserves order only with respect to the value $b$, thus, it is possible to construct an \textit{order preserving function at $b$}, with an output space having cardinality lesser than that of the input space. For example, the function: $F$ defined as : $F(x)=c_0<b$ $\forall x<b$, and $F(x)=c_1>b$ $\forall x<b$, is also an \textit{order preserving function at point $b$} having an output space of cardinality 2. 

\subsubsection{Formalizing Protocol \textsf{C}}
Protocol \textsf{C} uses 1 out of 2 Oblivious Transfer(OT) protcol as a subroutine. Some details of 1-out-of-2 OT protocol \citet{rabin2005exchange} have been discussed at the end of this subsection . 
\par Assume that Alice owns the private variable $a$, and Bob owns the private variable $b$. Let the bound on the length of the input be $d$ bits.
\begin{enumerate}
	\item Bob constructs the above mentioned function $F$ w.r.t. $b$. For generating the random values of $f_i$s using a single seed $S$, refer to Section \ref{sec:prg}
	\item Alice uses \textit{$1$-out-of-$2$ OT} for each input bit $a_i$ to obtain the bit encoding $f_i (a_i)$ for each $i$ and then computes $F(a) = F(a_n\ldots a_1)=f_n (a_n) + f_{n-1} (a_{n-1}) + \cdots + f_1(a_1)$.
	\item Bob then sends $F(b)$ to Alice, who then compares $F(a)$ and $F(b)$ and tells the result to Bob.
\end{enumerate}

\par The correctness of the protocol follows from the correctness of the construction of the desired function $F$ (w.r.t. $b$).

\par \textit{1-out-2 Oblivious Transfer} : For OT (1 out of 2), sender(S) has two secrets, $m_1$ and $m_2$, and would like the receiver(R) to receive one of them, as per her choice. But, R does not want S to know which secret it chooses and S wants R to know only the secret of her choice, revealing no information about the other secret. In our case, Bob has the secrets $f_i(0)$ and $f_i(1)$, and Alice has a bit $a_i$, corresponding to which she wants the secret $f_i(a_i)$. In the process, Alice wants her bit value to remain a secret and Bob wants Alice to obtain only the mapping corresponding to her bit value, hiding the mapping of the other bit.
\par Practically used 1-out-of-2 oblivious transfer involves four encryptions and two decryptions in a commutative encryption scheme, like RSA. Practically, public-key cryptography is expensive and hence it is preferrable to use a linear (even if large) number of cheap operations (in Oblivious Transfer). 

\subsubsection{Randomly Choosing $f_i$s using a single seed $S$} \label{sec:prg}
\par We first consider a PRG $G: \lbrace 0,1 \rbrace^c \longrightarrow \lbrace 0,1 \rbrace^{2dc}$, where $c$ is the bit size of the seed $S$ used in the generation of random number in $G$, and $d$ is the bound on bit-size of input.

\par We can then split up the output of $G$ into $2d$ pieces each of size $c$ and sequentially use the $2d$ pieces as seeds in a generator $G^\prime : \lbrace 0,1 \rbrace^c \longrightarrow \lbrace 0,1 \rbrace^{O(d)}$ to obtain the random values of $f_i$s in the desired ranges.

The security of this method has been argued in "Foundations of Cryptography-Basic Tools" by \citet{Goldreich01book}.

\subsubsection{Analysis}
\textsc{Communication Overhead:} The protocol uses $d$ parallel rounds of \textit{$1$-out-of-$2$ OT}, each on $O(d)$ bit numbers and requiring a different key. $O(d)$ bits are required in step 3 to communicate $F(b)$ to Alice. A single bit is transferred at the end to convey the result of the protocol (determined by Alice) to Bob. 

\textsc{Computation Complexity:} The protocol requires: $n$ additions (requiring a total of $O(n^2)$) + $2n$ $\times$ (Complexity of generating an $n$-bit random number). Since pseudorandom numbers can be generated efficiently, hence, the computations involved can be done efficiently.

\textsc{Security:}
Alice gets to see $f_i(a_i)$ for all $i$. However, since for each $i$ either $f_i(0)$ or $f_i(1)$ is chosen uniformly at random (according to our construction of F (Section \ref{ConsC})); and also, Alice doesn't know which of $f_i(0)$ or $f_i(1)$ is chosen at random, and neither does she know the range from which the parameters $m,n,l$ are chosen (which are also chosen uniformly at random), hence, it is difficult for her to deduce any information about the Bob's private variable $b$.
\par The security of the protocol also depends on the security of the PRG used and the security of the 1-out-of-2 OT protocol used to transfer $f_i(x_i)$s.

\section{Conclusion and Future Directions}\label{sec4}
\par It has been demonstrated that millionaires problem can be solved using linear communication overhead with an untrusted non-colluding third party and quadratic communication overhead without a third party. However, we still need to concretely prove a lower bound on the communication and computational complexity of solving the millionaire's problem, later extending it to the computation of general functions. A proof of lower bound would give us an insight in the difficulty of securely and jointly computing any function in general and hence would aid us in designing protocols for efficiently computing those functions.

\bibliographystyle{plain}

\end{document}